\documentclass[11pt,fleqn]{article}
\usepackage{fullpage}
\usepackage{ifpdf}
\ifpdf    
\usepackage{hyperref}
\else    
\usepackage[hypertex]{hyperref}
\fi
\hypersetup{colorlinks,allcolors=blue}

\usepackage{amsthm,amsmath,amssymb}
\usepackage[boxed]{algorithm}
\usepackage[noend]{algorithmic}
\usepackage{bbm}
\usepackage{graphicx}
\usepackage{xspace}
\usepackage{color}

\newtheorem{theorem}{Theorem}[section]
\newtheorem{lemma}[theorem]{Lemma}
\newtheorem{corollary}[theorem]{Corollary}
\newtheorem{definition}[theorem]{Definition}

\newtheorem{proposition}[theorem]{Proposition}

\newtheorem{claim}[theorem]{Claim}

\newcommand {\ignore} [1] {}
\def \ALG   {\mathsf{ALG}}
\def \Img   {\mathsf{Im}}
\def \Ker   {\mathsf{Ker}}
\def \dim   {\mathsf{dim}}
\newcommand{\R}{{\mathbb R}}
\DeclareMathOperator{\supp}{supp}

\newcommand{\eqdef}{:=}
\newcommand{\veps}{\omega}

\newcommand{\etal}{{\em et al.\ }\xspace}
\newcommand{\ceil}[1]{\lceil #1 \rceil}

\newcommand{\card}[1]{\lvert#1\rvert}
\newcommand{\norm}[1]{\lVert#1\rVert}

\renewcommand{\varrho}{\rho}
\newcommand{\tO}{\tilde{O}}

\title{Batch Sparse Recovery, \\ or How to Leverage the Average Sparsity}
\author{ Alexandr Andoni%
\thanks{Columbia University. 
Email: \texttt{andoni@cs.columbia.edu}.}
\and Lior Kamma%
\thanks{Aarhus University.
Part of this work was done while author was at Weizmann Institute of Science.
Work supported in part by the Israel Science Foundation grant \#897/13 and by a Villum Young Investigator Grant.
Email: \texttt{lior.kamma@cs.au.dk}. 
}
\and Robert Krauthgamer\thanks{Weizmann Institute of Science.
Work supported in part by the Israel Science Foundation grant \#897/13.
Email: \texttt{robert.krauthgamer@weizmann.ac.il}. 
}
\and Eric Price\thanks{The University of Texas at Austin.
Email: \texttt{ecprice@cs.utexas.edu}.}
}

\begin{document}
\maketitle

\thispagestyle{empty}
\setcounter{page}{0}

\begin{abstract}
We introduce a \emph{batch} version of sparse recovery, 
where the goal is to report a sequence of vectors $A_1',\ldots,A_m' \in \mathbb{R}^n$ 
that estimate unknown signals $A_1,\ldots,A_m \in \mathbb{R}^n$ 
using a few linear measurements, each involving exactly one signal vector, 
under an assumption of \emph{average sparsity}. 
More precisely, we want to have
\begin{equation} \label{eq:batchRec}
  \sum_{j \in [m]}{\|A_j- A_j'\|_p^p} 
  \le 
  C \cdot \min \Big\{ \sum_{j \in [m]}{\|A_j - A_j^*\|_p^p} \Big\} 
\end{equation}
for predetermined constants $C \ge 1$ and $p$, 
where the minimum is over all $A_1^*,\ldots,A_m^*\in\R^n$ 
that are $k$-sparse on average.
We assume $k$ is given as input, and ask for the minimal number 
of measurements required to satisfy~\eqref{eq:batchRec}.
The special case $m=1$ is known as stable sparse recovery and has been studied extensively. 

We resolve the question for $p =1$ up to polylogarithmic factors, 
by presenting a randomized adaptive scheme 
that performs $\tilde{O}(km)$ measurements 
and with high probability has output satisfying \eqref{eq:batchRec}, for arbitrarily small $C > 1$.
Finally, we show that adaptivity is necessary for every non-trivial scheme.
\end{abstract}

\newpage

\section{Introduction}

In \emph{sparse recovery} (or \emph{compressed sensing}) 
the goal is to reconstruct a signal vector $x\in\R^n$ 
using only \emph{linear measurements},
meaning that $x$ can be accessed only via queries of a linear form 
$x\mapsto a^\intercal x = \sum_i a_ix_i$.
To reduce the number of linear measurements, 
one usually assumes that the unknown signal $x\in\R^n$ is $k$-sparse
(defined as having at most $k$ non-zero entries, i.e., $\norm{x}_0\le k$),
or that $x$ is close to a $k$-sparse vector $x^*$,
and then the goal is to construct an estimate to $x$.
The astounding development of a concrete mathematical foundation for the problem 
by Cand\`es, Tao and Romberg~\cite{CRT06} and by Donoho~\cite{D06}, 
over a decade ago, has granted the problem huge attention, 
see, e.g., \cite{CW08, GI10, EK12, FR13} for exposition and references.  

Probably the most well-studied version of the problem, 
called \emph{stable sparse recovery}, is formulated as follows. 
A \emph{scheme} for dimension $n$ and sparsity bound $k\in[n]$, consists of 
(a) $t=t(n,k)$ 
linear measurements,
arranged as the rows of a \emph{sensing matrix} $S\in\R^{t \times n}$; and 
(b) a \emph{recovery algorithm} that uses the measurements vector $Sx$
to output $x'\in\R^n$.
Together, these should satisfy, for every signal $x\in\R^n$,
\begin{equation}
  \|x - x'\|_p \le C \min_{\text{$k$-sparse } x^*}\|x - x^*\|_p \;,
\label{eq:sparseRec}
\end{equation}
which is called an $\ell_p/\ell_p$ guarantee, 
where $C \ge 1$ and $p$ are some (predetermined) constants. 
The main goal is to minimize the number of measurements $t=t(n,k)$. 
For $C = \Theta(1)$, known schemes achieve 
the $\ell_1/\ell_1$ guarantee~\eqref{eq:sparseRec} 
using $t = O(k \log(n/k))$ measurements \cite{CRT06},
and this bound is asymptotically tight \cite{BIPW10}. 
A similar upper bound on $t$ is known also 
for the $\ell_2/\ell_2$ guarantee \cite{GLPS12} for random $S$, and where 
\eqref{eq:sparseRec} holds with high probability.

There are many other variants that focus on different considerations, such as constructing $S$ adaptively or achieving approximation factor $C$ arbitrarily close to $1$ at the cost of increasing $t$, see more details in Section~\ref{sec:related}.

\paragraph{Average Sparsity and Batch Recovery.} 

Although it is well established that many signal types are \emph{typically} sparse, 
a reasonable sparsity bound $k$ need not hold for {\em all signals},
and in some natural scenarios, 
a good upper bound might simply not be known in advance.
Consider, for example, $m$ servers in a large network, 
that communicate with a designated controller/coordinator server 
(the so-called {\em message-passing} model).
Denote the frequency vector of the requests made to server $j\in [m]$
by a column vector $A_j\in\R^n$.
To perform network analysis, such as anomaly detection and traffic engineering, 
the coordinator needs to examine information from all the $m$ servers, 
represented as a collective traffic matrix $A=(A_1,\ldots, A_m)\in\R^{n\times m}$.

The bottleneck in this setting is the vast amount of information, 
which exceeds communication constraints. 
Thus the coordinator usually collects only the most relevant data, 
such as the ``heavy'' entries from each server \cite{CQZXH14, Y14}.
Since typical traffic vectors have few heavy entries, 
it is more plausible to assume the columns have \emph{average sparsity} $k$,
than the significantly stricter assumption that every $A_j$ is $k$-sparse.

To formalize such a scenario as the problem of batch sparse recovery, 
we will need the following notation. 
We gather a sequence of column vectors $A_1, \ldots,A_m \in \R^n$
into an $n \times m$ matrix $A\eqdef (A_1, \ldots,A_m)$.
For $p \in [1, \infty)$, we let $\norm{A}_p$ denote the $\ell_p$-norm 
when $A$ is viewed as a ``flat'' vector of dimension $nm$,
i.e., $\norm{A}_p^p \eqdef \sum_{ij} |A_{ij}|^p = \sum_{j \in [m]} \norm{A_j}_p^p$,
for example, $p=2$ gives the Frobenius norm. 
Similarly, let $\|A\|_0 \eqdef \sum_{j \in [m]}{\|A_j\|_0}$ denote the sparsity of $A$,
i.e., the number of nonzero entries in $A$. 

\begin{definition} \label{def:batch}
In {\em batch recovery}, the input is a matrix $A\in\R^{n\times m}$ and 
a parameter $k \in [n]$. 
The goal is to perform linear measurements to columns of $A$, one column in each measurement, 
and then recover a matrix $A'$ satisfying 
$\norm{A - A'}_1 \le C \veps$
for some constant $C \ge 1$, where
\begin{equation*}
\veps = \min\limits_{\text{$(km)$-sparse } A^*}\|A - A^*\|_1 \;.
\end{equation*}
\end{definition}

A scheme for dimension $n$ and sparsity bound $k\in[n]$, consists therefore of 
constructing $m$ sensing matrices $S_1, \ldots, S_m$, where each $S_j \in\R^{t_j \times n}$ and 
a recovery algorithm that uses $S_1A_1,\ldots,S_mA_m$
to output $A'\in\R^{n \times m}$, such that $\|A - A'\|_1 \le C \veps$. The main goal is to minimize the total number of measurements $\sum_j{t_j}$.

\paragraph{Parallelization and Batch Recovery.}
One of the greatest features of linearity, and specifically of matrix-vector multiplication, is that one can perform linear measurements on parts of the input in parallel, and then combine the results.
In practice, when performing linear measurements $Sx$ on $x$, modern-day computers can delegate the computation to a many-core graphic 
processor (GPU). The GPU exploits this property of matrix-vector multiplication, and performs parts of the computation in many cores in parallel, rather than sequentially in the central processor (CPU).
Specifically, the GPU splits the $x$ into vectors $x_1,\ldots,x_m$ of smaller dimension and performs the measurements of each vector in a separate core \cite{FSH04, BGMV11, DB17}.
This model can be viewed as the aforementioned message-passing model, where the servers are the GPU cores, and the controller is the CPU.
If we assume $x \in \mathbb{R}^{nm}$ is approximately $km$-sparse, then $x_1,\ldots,x_m$ are approximately $k$-sparse {\em on average}. 

It might seem plausible, in this case, to treat $x_1,\ldots,x_m$ as one vector and invoke known sparse recovery algorithms. That is, construct a sensing matrix $S \in \mathbb{R}^{t \times mn}$, partition its columns into contiguous sets of size $n$ each, that is $S = \left(S_1 | S_2 | \ldots | S_m\right)$, where $S_j \in \mathbb{R}^{t \times n}$, multiply $S_1x_1,\ldots,S_mx_m$ in parallel (say in the GPU) and then analyze their concatenation $Sx$ in the CPU. The bottleneck in this setup is the bandwidth of the communication between the GPU and the CPU \cite{DAA+15}, since each of the $m$ cores needs to convey a vector of dimension $t = \tO(mk)$ (which is optimal for stable sparse recovery). The total communication consists of $\tO(m^2k)$ entries, which is prohibitive for the channel.

The setting of batch recovery allows more freedom in the choice of $S_1, \ldots, S_m$, as the number of measurements does not have to be identical for all sensing matrices. 
As we show in the next section, every non-trivial batch recovery algorithm must construct the sensing matrices adaptively. 
It should be noted that in practice, recent developments of GPUs allow for very good performance of adaptive algorithms employing GPUs as coprocessors \cite{LFB+12, LFDG14, CZ17}.

\subsection{Main Result} 
\label{sec:main}

Our main result presents an adaptive algorithm that recovers a sequence of $m$ 
vectors with average sparsity $k$ and uses at most $\tO(\varepsilon^{-3/2}km \log n)$ linear measurements to achieve $(1+\varepsilon)$ approximation. 
Our result and the ensuing discussion are stated in terms of the $\ell_1$ norm,
although all our results extend to the $\ell_2$ norm in a standard manner. 

\begin{theorem} \label{th:main}
There is a randomized adaptive scheme for batch recovery that, 
for every input $A,k$ and $\varepsilon$,
outputs a matrix $A'$ such that with high probability 
$\|A - A'\|_1 = (1+\varepsilon)\veps$,
where $\veps$ is the optimum as in Definition~\ref{def:batch}.
The algorithm performs $\tO(\varepsilon^{-3/2}k m \log n)$ linear measurements 
in $O(\log m)$ adaptive rounds.
\end{theorem}

\paragraph{Batch Recovery and Adaptivity.} 
A surprising and intriguing property that arises in 
batch sparse recovery is the necessity for adaptive algorithms. 
Intuitively, every algorithm for batch recovery must first learn how the heavy entries are distributed across the columns, 
{\em before} it can successfully reconstruct $A$.
Formally, let $A^*$ be a $(km)$-sparse matrix such that 
$\norm{A - A^*}_1 \le \veps$, 
and denote the sparsity of its $j$-th column by $k_j \eqdef \norm{A^*_j}_0$.
Once the values $\{k_j\}_{j \in [m]}$ are known, even approximately,
then the columns can be recovered near-optimally by
performing robust sparse recovery separately on each column. 

Indeed, our second result shows that this is more than mere intuition,
and proves that in the \emph{non-adaptive} setting, 
batch recovery is significantly harder than standard sparse recovery. 
Specifically, in Section~\ref{s:nonAdapt} we show that non-adaptive algorithms for reconstructing $A$ 
require $\Omega(nm)$ measurements, even in the noise-free case, as follows.

\begin{theorem} \label{t:nonAdaptiveLower}
For every $m,n$, every non-adaptive randomized scheme for batch recovery must make 
$\Omega(mn)$ linear measurements in the worst case, independently of $k$ 
even when $\veps=0$.
\end{theorem}

\subsection{Related Work}\label{sec:related}

\paragraph{Stable Sparse Recovery.}
Recall that a scheme for dimension $n$ and sparsity bound $k\in[n]$, consists of 
sensing matrix $S\in\R^{t \times n}$ and a recovery algorithm that receives $Sx$
and outputs $x'\in\R^n$ such that \eqref{eq:sparseRec} holds for some $C>1$.
This is often referred to as the {\em for all} model,
or sometimes a {\em uniform} or {\em deterministic guarantee}.
In contrast, in the {\em for each} model, 
the scheme (and in particular the matrix $S$) is random,
(drawn from a distribution designed by the algorithm),
and for every signal $x\in\R^n$ with high probability, 
the $\ell_p/\ell_p$ guarantee~\eqref{eq:sparseRec} 
holds.
The approximation factor $C$ can be made arbitrarily close to $1$ at the cost of increasing $t$ \cite{PW11}.
When the measurements may be constructed adaptively, 
there is a scheme with $t = O(k \log \log(n/k))$ \cite{IPW11}.

\paragraph{Matrix Reconstruction.}
The recovery of a matrix from partial or corrupted measurements 
has numerous applications in theoretical fields such as streaming 
and sublinear algorithms as well as practical ones, 
such as signal processing, communication-networks analysis, 
computer vision and machine learning. 
In many of these natural settings the matrix is typically sparse. 
Such settings include covariance matrices \cite{DSBN15}, 
adjacency matrices of sparse or random graphs \cite{M09,DSBN15}, 
image and video processing for facial recognition \cite{WYGSM09} 
and medical imaging \cite{M08}, 
in addition to traffic analysis of large communication networks \cite{CQZXH14}. 
Woodruff and Zhang \cite{WZ12, WZ13} considered 
a distributed model known as the {\em message-passing} model, 
similar to that described in the previous section. $m$ servers, 
each holding partial information  
regarding an unknown matrix $A$ (not necessarily a column, though),
need to communicate with a designated coordinator in order 
to compute some function of $A$. 
The communication they consider is not restricted to linear measurements.
They show communication complexity lower bounds 
in terms of bit-complexity for several designated 
functions (e.g. $\|A\|_0$ or $\|A\|_\infty$).

Considerable work has been made on the reconstruction of a matrix from
linear measurements performed on the matrix rather than on each column separately. 
That is, each measurement is of the form 
$A \mapsto B \bullet A = \sum{A_{ij}B_{ij}}$, 
where $B \in \mathbb{R}^{n \times m}$ is a sensing ``vector''. 
This model offers a much richer set of linear measurements, and, in fact, 
in terms of sparse recovery reduces matrix reconstruction to 
stable sparse recovery, albeit in some models, 
e.g. the message-passing model, such measurements are infeasible. 
Waters, Sankaranarayanan and Baraniuk \cite{WSB11} give an adaptive sensing 
algorithm for recovering a matrix which is the sum of a low-rank 
matrix and a sparse matrix using linear measurements on $A$. 

Dasarathy \etal \cite{DSBN15} recently considered matrix recovery 
from {\em bilinear} measurements (also called  {\em tensor products}),
i.e., measurements of the form $A \mapsto v^tAu = \sum{A_{ij}v_iu_j}$, 
where $v,u$ are sensing vectors. 
This model can be viewed as a restriction of the aforementioned
model, in which every sensing vector $B$ is a matrix of rank $1$. 
They show how to reconstruct a sparse matrix using bilinear measurements, 
in the special case where the heavy entries of $A$ are not concentrated
in a few columns.

\section{Preliminaries}
For every $x \in \mathbb{R}^n$ let $\supp(x) \eqdef \{j \in [n]: x_j \ne 0\}$ be the set of non-zero entries of $x$, and denote $\|x\|_0 = | \supp(x) |$. 
\begin{definition}
Let $u,v \in \mathbb{R}^n$. We say that $u$ {\em agrees} with $v$ if $\supp(u) \subseteq \supp(v)$, and in addition, for every $j \in \supp(u)$, $u_j = v_j$.
\end{definition}
Note that the relation defined above is not symmetric. The following straightforward claim demonstrates some fundamental properties of the agreement relation.
\begin{claim}
Let $u,v,w \in \mathbb{R}^n$, and assume $u$ agrees with $v$, then
\begin{enumerate}
	\item $v-u$ agrees with $v$;
	\item $\|u\|_1 \le \|v\|_1$ ; and 
	\item if $w$ agrees with $v-u$, then $u+w$ agrees with $v$.
\end{enumerate}
\end{claim}

It is well known \cite{PW11} that given a sparsity parameter $s$ and some $\varepsilon>0$, 
there exist sensing matrices and associated recovery algorithms for the $(1+\varepsilon)$-approximate 
stable sparse recovery problem such that the number of linear measurements is $\tO(\varepsilon^{-1/2}s \log n)$.
\begin{theorem}[\cite{PW11}]\label{th:compNoise}
Let $s \in [n]$ and $\varepsilon>0$, there exists a random sensing matrix $S \in \mathbb{R}^{t \times n}$ 
for $t=O\left(\frac{\log^3 1/\varepsilon}{\sqrt{\varepsilon}}s \log n\right)$, and an associated recovery algorithm 
such that for every $x \in \mathbb{R}^n$, given $Sx$, produces  
a vector $x'(s)$ that satisfies $\|x - x'(s)\|_1 \le (1+\varepsilon) \min\limits_{x^* \;\; s-sparse}\|x - x^*\|_1$.
Moreover, such $S$ can be found efficiently.
\end{theorem} 

Our algorithms use this result of \cite{PW11} 
as a step in the construction of an estimate. 
We note that the choice of a specific recovery algorithm is 
not crucial. 
The role of Theorem~\ref{th:compNoise} in our algorithms 
can be replaced by {\em any} (adaptive or non-adaptive) 
stable sparse recovery guarantee.
Thus, for example, Theorem~\ref{th:compNoise} can be replaced 
by an adaptive sparse recovery result (e.g. \cite{IPW11}), thus 
performing less measurements at the cost of increasing the number of 
adaptive rounds.

Indyk \cite{I06} showed that one can construct a matrix 
$S \in \mathbb{R}^{t \times n}$, where $t =  O\left(\log n\right)$ 
such that for every $x \in \mathbb{R}^n$, given $Sx$, we can estimate $\|x\|_1$ 
up to a constant factor with high probability.
\begin{theorem}[\cite{I06}]\label{th:estimateNorm}
There exists a sensing matrix $S \in \mathbb{R}^{t \times n}$ for $t=O\left(\log n \right)$, 
and an associated algorithm such that for every $x \in \mathbb{R}^n$, given $Sx$, 
produces a number $\varrho(x)$ that satisfies $\frac{1}{2}\|x\|_1 \le \varrho(x) \le 2\|x\|_1$ 
with probability at least $1 - \frac{1}{n^{\Omega(1)}}$.
Moreover, such $S$ can be found efficiently. 
\end{theorem}

\section{Constant-Approximation Batch Reconstruction} \label{sec:mainWeakProof}

In this section we prove a slightly weaker version of Theorem~\ref{th:main}, by presenting an iterative algorithm that approximates an unknown sequence of signals up to some constant factor by an on-average sparse sequence.
Formally, we show the following.

\begin{theorem} \label{th:mainWeak}
There exists a constant $C>1$ such that there is a randomized adaptive scheme for batch recovery that, 
for every input $A$ and $k$,
outputs a matrix $A'$ such that with high probability 
$\|A - A'\|_1 \le C\veps$,
where $\veps$ is the optimum as in Definition~\ref{def:batch}.
The algorithm performs $O(k m \log n \log m)$ linear measurements 
in $O(\log m)$ adaptive rounds, 
and its output $A'$ is $km$-sparse.
\end{theorem}

The algorithm presented in this section, which is described in detail as Algorithm~\ref{alg:noiseWeak} will be used as a "preprocessing" step in the proof of Theorem~\ref{th:main} in Section~\ref{sec:mainProof}.

The algorithm performs $\log (2m)$ iterations. Throughout the execution, it maintains a set $I \subseteq [m]$, of the indices of all columns not yet fixed by the algorithm. Initially $I$ is the entire set $[m]$. At every iteration $\ell \in [\log (2m)]$, the algorithm performs standard sparse recovery on each of the columns indexed by $I$ with sparsity parameter $2^{\ell+1}k$ and $\varepsilon=2$ (any constant will do here), and constructs vectors $\{A_j^{tmp}\}_{j \in I}$. Applying Theorem~\ref{th:estimateNorm}, the algorithm additionally estimates the residual error $\|A_j - A_j^{tmp}\|_1$ for all $j \in I$, and then chooses the $\frac{1}{2}|I|$ indices for which the residual error is smallest. For each such index $j$, the algorithm fixes $A_j^{alg}$ to be $A_j^{tmp}$, and removes $j$ from $I$. 
After $\log (2m)$ iterations, 
the algorithm returns a ``truncated'' version of the matrix $A^{alg}$, 
namely, a $km$-sparse matrix $A^{fin}$, whose non-zeros are 
simply the $km$ heaviest entries (largest absolute value) in $A^{alg}$. 

\begin{algorithm}[th]
\begin{algorithmic}[1]
\STATE initialize $I \leftarrow [m]$.
\FOR{$\ell = 1$ \TO $\log (2m)$} \label{l:mainStartWeak}
\FORALL{$j \in I$}
\STATE let $A_j^{tmp} \leftarrow A_j'(2^{\ell+1}k)$ \COMMENT{by applying Theorem~\ref{th:compNoise} with $\varepsilon=2$}
\STATE $\rho_j \leftarrow \rho(A_j - A_j^{tmp})$ \COMMENT{by applying Theorem~\ref{th:estimateNorm}}
\STATE let $I_\ell \subseteq I$ be the set of $\left\lceil m/2^{\ell} \right\rceil$ indices $j \in I$ with smallest $\rho_j$.
\STATE let $I \leftarrow I \setminus I_\ell$
\FORALL{$j \in I_\ell$} 
\STATE let $A_j^{alg} \leftarrow A_j^{tmp}$ \COMMENT{fix the columns $\{ A_j^{tmp} : j \in I_{\ell} \}$}
\ENDFOR 
\ENDFOR
\ENDFOR \label{l:mainEndWeak}
\STATE let $A^{fin}$ be the $km$-sparse matrix whose non-zeros entries are the $km$ heaviest entries of $A^{alg}$.
\RETURN $A^{fin}$.
\caption{Algorithm for Batch Reconstruction}
\label{alg:noiseWeak}
\end{algorithmic}
\end{algorithm}

Prior to analyzing Algorithm~\ref{alg:noiseWeak} in the next section, 
let us note that we can easily bound the number of times it invokes Theorems \ref{th:estimateNorm} and \ref{th:compNoise}, and show that with probability at least $1 - \frac{1}{n^{\Omega(1)}}$ all invocations succeed. We therefore condition on that event. 
For sake of simplicity, we assume that $m$ is a power of $2$.
The next claim follows by simple induction.
\begin{proposition}
For every $\ell \in [\log (2m)]$, at the beginning of the $\ell$th iteration, $|I| = \tfrac{m}{2^{\ell-1}}$.
\end{proposition}
It follows that at the end of the last iteration of the main loop, $I = \emptyset$, and thus the output columns are all well-defined.

\subsection{Controlling the Noise and the Number of Measurements}

The main challenge is to bound the relative error by $O(\veps)$.
While the algorithm may seem very natural, the straightforward analysis incurs an extra factor of $\log m$ on the approximation factor, as follows.
By averaging, for every $\ell \in [\log (2m)]$, 
at most $\tfrac{m}{2^{\ell+1}}$ columns $A^*_j$ have 
$\|A^*_j\|_0 > 2^{\ell+1}k$, 
and at most $\tfrac{m}{2^{\ell+1}}$ columns $A^*_j$ have 
$\|A_j - A^*_j\|_1 > \tfrac{2^{\ell+1}}{m} \veps$. 
The total number of columns in these two groups is at most 
$2\cdot \tfrac{m}{2^{\ell+1}} = \tfrac12 \card{I}$, 
and thus at least $\ceil{\tfrac12 \card{I}} = \card{I_\ell}$ columns in $I$ 
are not in these two groups.
By the sparse recovery guarantees on these columns and the choice of $I_\ell$,
\begin{equation}
  \forall j \in I_\ell, \qquad
  \|A_j - A_j^{alg}\|_1 \le \frac{3 \cdot 2^{\ell+1}}{m} \veps .
\label{eq:simpleBound}
\end{equation}
Summing these over all values of $\ell$ we get that 
$\|A - A^{alg}\|_1 \le O( \veps \log m)$.

In order to improve this guarantee to $O(\veps)$, 
we need to use a more subtle argument, and tighten the bound in~\eqref{eq:simpleBound}.
We replace the term $\|A\|_1/m$, which represents the norm of an average column, 
with the norm of a specific column, 
and the crux is that they sum up (over all iterations $\ell$) very nicely,
because these summands correspond (essentially) to distinct columns. 

We additionally note that it actually suffices to prove that at the end of the execution, 
$\|A - A^{alg}\|_1 \le O(\veps)$. 
To see this, let $A^*$ denote a $(km)$-sparse matrix satisfying 
$\|A - A^*\|_1 = \veps$. Since $A^{fin}$ is the $km$-sparse matrix closest to $A^{alg}$, then $\|A^{alg} - A^{fin} \|_1 \le \|A^{alg} - A^* \|_1$.
Using these bounds and (twice) the triangle inequality, we get
$$ 
  \|A - A^{fin}\|_1 
  \le \| A - A^{alg} \|_1 + \| A^{alg} - A^{fin} \|
  \le 2\|A - A^{alg}\|_1 + \veps \;.
$$

We start the analysis by first bounding the number of linear measurements performed by the algorithm.
\begin{lemma}\label{l:boundPerIter}
During the $\ell$th iteration of the main loop, Algorithm~\ref{alg:noiseWeak} performs $O(km \log n)$ linear measurements. 
\end{lemma}
\begin{proof} 
At the beginning of the $\ell$th iteration, $|I| \le \tfrac{m}{2^{\ell-1}}$. For every $j \in I$, the algorithm performs $O(2^{\ell + 1} k \log n)$ linear measurements on $A_j$. 
Thus the total number of measurements performed during  the $\ell$th iteration is $O(2^{\ell + 1} k \log n)\cdot |I| \le O(mk \log n)$. 
\end{proof}
The algorithm performs $\log (2m)$ iterations, and thus 
the total number of linear measurements performed throughout the execution is at most $O(km \log n \log m)$.
It remains to show that $\|A - A^{alg}\|_1 \le O(\veps)$. 

To this end, denote for every $j \in [m]$, $\veps_j = \|A_j - A^*_j\|_1$, and note that $\sum_{j \in [m]}{\veps_j} = \veps$. Let $\veps_{(1)} \ge \veps_{(2)} \ge \ldots \ge \veps_{(m)}$ be a non-increasing ordering of $\{\veps_j\}_{j \in [m]}$.
\begin{lemma}
For every $\ell \in [\log (m/2)]$ and $j \in I_\ell$, $\rho_j \le 6 \veps_{(2^{-\ell-1}m)}$.
\end{lemma}
\begin{proof}
Fix some $\ell \in [ \log (m/2) ]$, and consider the set $I$ at the beginning of the $\ell$th iteration. 
Recall that $|I| = m/2^{\ell-1}$, and $I_\ell$ is the set of $m/2^{\ell}$ indices $j \in I$ with smallest $\rho_j$. Observe that to prove the claim, it is enough to show that 
$$\Pr_{j \in I}\left[ \rho_j > 6 \veps_{(2^{-\ell-1}m)} \right] \le \frac{1}{2} \;.$$
To this end, consider an arbitrary $j \in I$ with $\rho_j > 6 \veps_{(2^{-\ell-1}m)}$. If, in addition, $\|A_j^*\|_0 \le 2^{\ell+1}k$, then by Theorem~\ref{th:compNoise}
$$\|A_j - A_j^{tmp}\|_1 = \|A_j - A_j'(2^{\ell + 1}k)\|_1 \le 3 \min\limits_{x^* \;\; 2^{\ell+1}k-sparse}\|A_j - x^*\|_1 \le 3 \|A_j - A_j^*\|_1 = 3\veps_j\;.$$
By Theorem~\ref{th:estimateNorm}, $\|A_j - A_j^{tmp}\|_1 \ge \tfrac{1}{2}\rho(A_j - A_j^{tmp})=\tfrac{1}{2}\rho_j$, and therefore 
$$\veps_j \ge \frac{1}{3} \|A_j - A_j^{tmp}\|_1  \ge \frac{1}{6}\rho_j > \veps_{(2^{-\ell-1}m)}\;.$$ We conclude that for every $j \in I$, if 
$\rho_j > 6 \veps_{(2^{-\ell-1}m)}$, then either $\|A_j^*\|_0 > 2^{\ell+1}k$ or $\veps_j > \veps_{(2^{-\ell-1}m)}$.
By definition, $\veps_j > \veps_{(2^{-\ell-1}m)}$ occurs for at most $\frac{m}{2^{\ell+1}}$ indices $j \in [m]$.
In addition, since $\mathbb{E}_{j \in [m]}[\|A^*_j\|_0] = k$, then at most $\frac{m}{2^{\ell+1}}$ indices $j \in [m]$ satisfy $\|A_j^*\|_0 > 2^{\ell+1}k$. Thus at most $\tfrac{m}{2^{\ell}} = \tfrac{1}{2}|I|$ indices $j \in [m]$ satisfy $\rho_j > 6 \veps_{(2^{-\ell-1}m)}$.
The claim follows.
\end{proof}
\begin{lemma}
For every $\ell \in [\log (m/2)]$, $$\sum_{j \in I_\ell}{\|A_j - A_j^{alg}\|_1} \le 24 \sum_{j \in \left[\frac{m}{2^{\ell+1}}, \frac{m}{2^{\ell}} -1\right]}{\veps_{(j)}} \;.$$
\end{lemma}
\begin{proof}
Fix some $\ell \in [\log (m/2)]$, and let $j \in I_\ell$. Then $A_j^{alg}$ was fixed in the $\ell$th iteration, and thus by the previous claim, $\|A_j - A_j^{alg}\|_1 \le 2 \rho(A_j - A_j^{alg}) = 2 \rho_j \le 12 \veps_{(2^{-\ell-1}m)}$. Since $|I_\ell| = \tfrac{m}{2^\ell}$ and $\{\veps_{(j)}\}_{j \in [m]}$ is non-increasing
$$\sum_{j \in I_\ell}{\|A_j - A_j^{alg}\|_1} \le 12 \cdot \tfrac{m}{2^\ell} \cdot \veps_{(2^{-\ell-1}m)} \le 12 \cdot 2 \sum_{j \in \left[\frac{m}{2^{\ell+1}}, \frac{m}{2^{\ell}} - 1\right]}{\veps_{(j)}} \;.$$
\end{proof}
\begin{corollary} \label{cor:noiseBound}
$\|A - A^{alg}\|_1 \le O(\veps)$.
\end{corollary}
\begin{proof}
Since $\bigcup_{\ell \in [\log (2m)]}I_\ell = [m]$, and the sets $\{I_\ell\}_{\ell \in \log (2m)}$ are pairwise disjoint, 
then $$\|A - A^{alg}\|_1 = \sum_{\ell = 1}^{\log (2m)}{\sum_{j \in I_\ell}{\|A_j - A_j^{alg}\|_1}} = \sum_{j \in I_{\log (2m)}\cup I_{\log m}}{\|A_j - A_j^{alg}\|_1} + \sum_{\ell = 1}^{\log m}{\sum_{j \in I_\ell}{\|A_j - A_j^{alg}\|_1}} \;.$$
Observe first that $|I_{\log (2m)} \cup I_{\log m}| = 3$.  Moreover, for every $j \in I_{\log (2m)} \cup I_{\log m}$, $A_{j}^{alg}$ is constructed by applying standard sparse recovery on $A_j$ with sparsity parameter $ \ge mk$. Since $\|A^*_{j}\|_0 \le km$, then similarly to the previous proof we get that $\|A_{j} - A^{alg}_{j}\|_1 \le 3 \veps_{j}  \le 3 \veps $.
Therefore by the previous lemma 
\begin{equation*}
\|A - A^{alg}\|_1 \le 9 \veps  + 24\sum_{\ell = 1}^{\log m}{\sum_{j \in \left[\frac{m}{2^{\ell+1}}, \frac{m}{2^{\ell}} - 1\right]}{\veps_{(j)}}} \le 9 \veps  + 24\sum_{j=1}^{m}{\veps_{(j)}} = O(\veps) \;.
\end{equation*}
\end{proof}
Theorem~\ref{th:mainWeak} follows from Lemma~\ref{l:boundPerIter} and Corollary~\ref{cor:noiseBound}.

\section{$(1+\varepsilon)$-Approximation Batch Reconstruction} \label{sec:mainProof}

In this section we present an algorithm that, given access
to $A_1,\ldots,A_m$ via linear measurements, in addition to a parameter
$k \in \mathbb{N}$ and $\varepsilon > 0$, constructs
vectors $A_1^{alg},\ldots,A_m^{alg}$ satisfying $\|A - A^{alg}\|_1 \le (1+\varepsilon)\veps$. 
We will additionally show that the algorithm performs a total of at most $\tilde{O}(\varepsilon^{-3/2}km
\log n)$ linear measurements. 

\subsection{High Level Description}
In the beginning of the execution, for every $j \in [m]$, the algorithm invokes Algorithm~\ref{alg:noiseWeak} on $A$ in order to construct a matrix $A^{init}$ satisfying $\|A - A^{init} \|_1 \le O(1) \cdot \veps$. For every $j \in [m]$, the vector $D_j$ is shorthand for $A_j - A_j^{init}$ the vector of residual entries of $A_j$, not yet recovered by the algorithm. As previously noted, since we can access $A_j$ via linear queries, we can perform linear measurements on $A_j - v$ for every known $v \in \mathbb{R}^n$. Specifically, we can perform linear measurements on $D_j$.

Next, the algorithm constructs new vectors $A_1^{tmp}, \ldots, A_m^{tmp}$, where each $A_j^{tmp}$ is essentially an estimation of the heaviest entries of $D_j$, as in Theorem~\ref{th:compNoise}. We note that this step uses the construction algorithm in Theorem~\ref{th:compNoise} as ``black box'', and any other stable sparse recovery result, can be used here. 
The exact number of entries in question is determined by the algorithm for each $j \in [m]$ separately, as will be described shortly. For every $j \in [m]$, $A_j^{alg}$ is then defined as $A_j^{init} + A_j^{tmp}$, and the algorithm returns $A^{alg}_1,\ldots,A^{alg}_m$. By carefully constructing $A_1^{tmp}, \ldots, A_m^{tmp}$ we show that the quantity $\|A - A^{alg}\|_1$ can be made arbitrarily close to $\veps$.

\paragraph{Constructing $A^{tmp}$ .} 
The key challenge is constructing $A_1^{tmp}, \ldots, A_m^{tmp}$ 
while not exceeding the budget, i.e. number of linear measurements too much. 
For this purpose we use a subtle bucketing on the columns 
with respect to the estimated $\ell_1$ norm of each column. 
The algorithm then ``invests'' an appropriate amount of measurements 
for each bucket, divided equally between the columns mapped to the respective bucket. 

In order to gain intuition, let us consider the special case of the problem where $\|A_j - A_j^{init}\|_1 = 1$ for all $j \in [m]$, 
thus initially $m = \|A - A^{init}\|_1 = C\veps$ for some $C>1$. 
Let $A^*$ denote a $(km)$-sparse matrix satisfying $\|A - A^*\|_1 \le \veps$. 
Let $J = \{j \in [m] : \|A^*_j\|_0 \le k/\varepsilon\}$ be the set of $(k/\varepsilon)$-sparse columns of $A^*$. Since $A^*$ is $(km)$-sparse, 
a straightforward averaging argument yields that $|J| \ge (1-\varepsilon)m$. 
Following Theorem~\ref{th:compNoise}, by performing $\tilde{O}(\varepsilon^{-3/2}k \log n)$ linear measurements 
on each column, we can find vectors $A'_1,\ldots,A'_m$ satisfying 
$\|A_j - A'_j\|_1 \le (1+\varepsilon) \veps_j$, 
where $\veps_j = \min\limits_{\hat{A}_j \; (k/\varepsilon)-sparse}\|A_j - \hat{A}_j\|_1$. 
Note that in particular this means that for every $j \in J$, $\veps_j \le \|A_j - A^*_j\|_1$. 
Therefore, by letting $A^{alg}_j = A'_j$ for all $j \in J$ and $A^{alg}_j = A^{init}_j$ otherwise, we get
\begin{equation*}
\begin{split}
\|A- A^{alg}\| =  \sum_{j \in J}{\|A_j - A'_j\|_1} + \sum_{j \notin J}{\|A_j - A^{init}_j\|_1} 
 \le (1+\varepsilon) \sum_{j \in J}{\|A_j - A^*_j\|_1} + \varepsilon m \le (1+O(\varepsilon))\veps \;,
\end{split}
\end{equation*}
where the inequality before last is due to the fact that 
$\|A_j - A_j^{init}\|_1 = 1$ for all $j \in [m]$ and $|J| \ge (1-\varepsilon)m$, and the last inequality follows since $m = O(\veps)$.

In general, however, we cannot assume that all columns have the same $\ell_1$ norm (note that the value $1$ in the example above is arbitrary), even up to a constant. 
In order to implement the demonstrated approach we bucket the columns in such a way that columns that are in the same bucket have approximately the same $\ell_1$ norm. The algorithm performs a total of $\tilde{O}(\varepsilon^{-3/2}mk \log n)$ measurements in each bucket. The number of measurements is equally divided by the columns in each bucket. In addition, by employing a subtle charging scheme, we show that the number of buckets is not too large (in fact, it is at most $O(\log (m/\varepsilon))$).
The algorithm is described in detail as Algorithm~\ref{alg:noise}.

\begin{algorithm}[th]
\begin{algorithmic}[1]
\FORALL{$j \in [m]$}
\STATE let $A^{init}$ be the result of applying Algorithm~\ref{alg:noiseWeak} on $A$. 
\STATE $D_j \leftarrow A_j - A^{init}_j$ \hspace{2pc} \COMMENT{implicitly}
\ENDFOR
\STATE let $M = \max_{j \in [m]}\varrho(D_j)$.
\FORALL{$i \in [\log (m/\varepsilon)]$}
\STATE let ${\cal E}_i \leftarrow \{j \in [m] : 2^{-i}M < \varrho(D_j) \le 2^{-i+1}M  \}$.
\ENDFOR
\FORALL{$j \in \bigcup_i{\cal E}_i$}
\STATE let $w(j)$ be the unique $i$ such that $j \in {\cal E}_i$. \label{l:unique}
\STATE let $A^{tmp}_j \leftarrow \left(D^{init}_j\right)'\left(\frac{100mk}{\varepsilon|{\cal E}_{w(j)}|}\right)$ as in Theorem~\ref{th:compNoise}. \label{l:construct}
\STATE let $A^{alg}_j \leftarrow A^{init}_j + A^{tmp}_j$
\ENDFOR
\RETURN $A^{alg}_1, \ldots, A^{alg}_n$.
\caption{Batch Reconstruction}
\label{alg:noise}
\end{algorithmic}
\end{algorithm}

We note that by adding one more adaptive round to the constructions of $A^{init}$ and $A^{tmp}$, we may assume without loss of generality that in addition $A^{init}$ agrees with $A$ and $A^{tmp}$ agrees with $D$. 

\subsection{Controlling the Noise and the Number of Measurements}
This section is devoted to the proof of Theorem~\ref{th:main}. We first note, that we can bound the number of times the algorithm invokes Theorem~\ref{th:estimateNorm} in order estimate the $\ell_1$ norm of a vector, and show that with probability at least $1 - \frac{1}{n^{\Omega(1)}}$ all invocations succeed. We therefore condition on that event.

We can now turn to analyze the algorithm and prove Theorem~\ref{th:main}. 
Let ${\cal H} \eqdef \{j \in [m] : \varrho(D_j) > \varepsilon M/m\}$ denote the set of the ``heavier'' columns of $D$. Whenever $j \notin {\cal H}$, the algorithm sets $A_j^{alg}$ to be simply $A_j^{init}$. For every $j \in {\cal H}$, the algorithm (line \ref{l:construct}) sets $A_j^{tmp}$ to be $D_j'\left(\frac{100mk}{\varepsilon|{\cal E}_{w(j)}|}\right)$, and sets $A_j^{alg} = A_j^{init} + A_j^{tmp}$.
The following claim gives a bound on the number of linear measurements performed after initialization.

\begin{claim}\label{c:boundPerIter}
Algorithm~\ref{alg:noise} performs $\tilde{O}(\varepsilon^{-3/2}mk\log n)$ linear measurements.
\end{claim}
\begin{proof} 
Let $i \in [\log (m/\varepsilon)]$, and let $j \in {\cal E}_i$. Then $A_j^{tmp} \eqdef D'_j\left(\frac{100mk}{\varepsilon|{\cal E}_i|}\right)$. By Theorem~\ref{th:compNoise} the number of linear queries performed to construct $A_j^{tmp}$ is $O\left(\varepsilon^{-3/2}\frac{mk}{|{\cal E}_i|} \log^3(1/\varepsilon)\log n \right)$. Therefore the total number of queries performed on all $j \in {\cal E}_i$ is at most $O\left(\varepsilon^{-3/2}mk \log^3(1/\varepsilon)\log n \right)$, and thus, the total number of linear queries performed to construct $\{A_j^{tmp}\}_{j \in {\cal H}}$ is $O(\varepsilon^{-3/2}mk \log^3(1/\varepsilon)\log n \log (m/\varepsilon))$. 
\end{proof}
The rest of this section is devoted to the proof of the following lemma, which implies Theorem~\ref{th:main}.
\begin{lemma} \label{l:1PlusEps}
$\|A-A^{alg}\|_1 \le (1+O(\varepsilon))\veps$.
\end{lemma}

We first show that the columns that the algorithm chooses not to ``invest'' in do not contribute much to $\|A - A^{alg}\|_1$. Intuitively, if the algorithm does not perform measurements on a column, that column is ``close'' to the corresponding column in $A$. 
By the definition of ${\cal H}$ we get that for every $j \notin {\cal H}$, $\|A_j-A^{alg}_j\|_1 = \|D_j\|_1 \le 2\varrho(D_j) \le 2\varepsilon M/m$. Therefore
$$\sum_{j \notin {\cal H}}{\|A-A^{alg}\|_1} \le m \cdot 2\varepsilon M/m \le  O(\varepsilon) \veps \;,$$
and thus 
$$\|A-A^{alg}\|_1 = \sum_{j \in [m]\setminus {\cal H}}{\|A-A^{alg}\|_1} + \sum_{j \in {\cal H}}{\|A-A^{alg}\|_1} \le  O(\varepsilon) \veps + \sum_{j \in {\cal H}}{\|A-A^{alg}\|_1}\;.$$
In what follows we show that there are $mk$-entries in $D$ such that if we replace all the values in these entries by zeroes, then the $\ell_1$ norm of the resulting matrix is at most $\veps$. 
Let $A^*$ denote a $(km)$-sparse matrix satisfying $\|A - A^*\|_1 \le \veps $. Without loss of generality we may assume that $A^*$ agrees with $A$.
Let $A^{**}$ be an $n \times m$ matrix composed of exactly those non-zero entries of $A^*$ which are still equal to $0$ in $A^{\ell}$. All other entries of $A^{**}$ are zeroes. Formally, for every $j \in [m]$, $\supp(A^{**}_j) = \supp(A_j^*) \setminus \supp(A_j^{init})$, and for every $i \in \supp(A_j^{**})$, $A^{**}_{ij} = A^*_{ij}$.

\begin{claim}\label{c:A**}
$A^{**}$ is $(km)$-sparse and in addition $\|D_j- A^{**}_j\|_1 \le \|A_j - A^*_j\|_1$ for all $j \in [m]$. 
\end{claim}

\begin{proof}
By definition $A^{**}$ agrees with $A^{*}$, and thus $\|A^{**}\|_0 \le \|A^*\|_0 = km$. 
Next, fix some $j \in [m]$. Then $D_j - A^{**}_j = A_j - A_j^{init} - A^{**}_j$. Since $A_j^{*}$ agrees with $A_j$, we get that $A_j^{**}$ agrees with $A_j$. It follows that $A^{**}_j$ and $A_j^{init}$ both agree with $A_j$ and in addition, $\supp(A^{**}_j) \cap \supp(A^\ell_j) = \emptyset$. Therefore $A_j^{init} + A^{**}_j$ agrees with $A_j$. Moreover, 
$$\supp(A_j^*) \subseteq (\supp(A^{**}_j) \cup \supp(A^{init}_j)) = \supp(A^{**}_j + A^{init}_j) \subseteq \supp(A_j)\;.$$ Therefore $A_j^*$ agrees with $A_j^{init} + A^{**}_j$, and thus 
$$\|D_j - A_j^{**}\|_1 = \|A_j - A_j^{init} - A^{**}_j\|_1 = \|A_j - (A_j^{init} + A^{**}_j)\|_1 \le \|A_j - A_j^*\|_1 \;.$$
\end{proof}

In order to bound $\sum_{j \in {\cal H}}{\|A-A^{alg}\|_1}$ we define the following subsets of ${\cal H}$.
$${\cal X} \eqdef \left\{j \in {\cal H} : \frac{100km}{\varepsilon|{\cal E}_{w(j)}|} \ge \|A^{**}_j\|_0 \right\}$$
$${\cal Y}_i \eqdef \left\{j \in {\cal E}_i : \frac{100km}{\varepsilon|{\cal E}_i|} < \|A^{**}_j\|_0 \right\} ,\quad \quad i \in [\log (m/\varepsilon)] \;\;.$$
Since $\bigcup_{i \in [\log (m/\varepsilon)]}{\cal E}_i = {\cal H}$, we conclude that
\begin{equation}
\sum_{j \in {\cal H}}{\|D_j^\ell - A_j^{tmp}\|_1} = \sum_{j \in {\cal X}}{\| D_j^\ell - A_j^{tmp}\|_1} + \sum_{i \in [\log (m/\varepsilon)]}{\sum_{j \in {\cal Y}_i}{\| D_j^\ell - A_j^{tmp}\|_1}}
\label{eq:partitionSparsity}
\end{equation}
\begin{claim}\label{c:sumDense}
$\sum_{j \in {\cal X}}{\|A_j - A_j^{alg}\|_1} \le (1+\varepsilon)\veps$
\end{claim}
\begin{proof}
Let $j \in {\cal X}$, and denote $k_j = \frac{100km}{\varepsilon|{\cal E}_{w(j)}|}$. Then $A_j^{tmp} = D'\left(k_j\right)$ and thus by Theorem~\ref{th:compNoise} 
$$\|D_j - A_j^{tmp}\|_1 \le (1+\varepsilon)\min\limits_{x^* \;\; k_j-sparse}\|D_j - x^*\|_1 \le  (1+\varepsilon) \|D_j - A_j^{**}\|_1 \le (1+\varepsilon)\|A_j - A^*_j\|_1 \;,$$ 
where the inequality before last is due to the fact that $j \in {\cal X}$ and therefore $A_j^{**}$ is $k_j$-sparse. Summing over all $j \in {\cal X}$ we get 
$$\sum_{j \in {\cal X}}{\|D_j - A_j^{tmp}\|_1} \le (1+\varepsilon) \sum_{j \in [m]}{\|A_j - A_j^{*}\|_1} \le (1+\varepsilon) \veps  \;.$$
\end{proof}

\begin{claim}\label{c:sumSparse}
$\sum_{i \in [\log (m/\varepsilon)]}{\sum_{j \in {\cal Y}_i}{\|D_j - A_j^{tmp}\|_1}} \le O(\varepsilon)\veps$
\end{claim}
\begin{proof}
Fix $i \in [\log (m/\varepsilon)]$.
For every $j \in {\cal Y}_i \subseteq {\cal E}_i$, $w(j)=i$. Since $A^{**}$ is $(km)$-sparse, $\mathbb{E}_{j \in {\cal E}_i}[\|A_j^{**}\|_0] \le \frac{km}{|{\cal E}_i|}$. By Markov's inequality, $$\Pr_{j \in {\cal E}_i}\left[j \in {\cal Y}_i\right] = \Pr_{j \in {\cal E}_i}\left[\|A_j^{**}\|_0 > \frac{100 km}{\varepsilon|{\cal E}_i|}\right] \le \frac{\varepsilon}{100} \;.$$
Since, in addition, $A_j^{tmp}$ agrees with $D_j$ we get that
$$
\sum_{j \in {\cal Y}_i}{\| D_j - A_j^{tmp}\|_1} \le \sum_{j \in {\cal Y}_i}{\| D_j\|_1} \le 2\sum_{j \in {\cal Y}_i}{\varrho(D_j)} \le 2\sum_{j \in {\cal Y}_i}{2^{-i+1}M} = 4 \cdot 2^{-i}M|{\cal Y}_i| \;,$$
where the last inequality is due to the fact that ${\cal Y}_i \subseteq {\cal E}_i$.
Note that $\sum_{j \in {\cal E}_i}{\varrho(D_j)} > 2^{-i}M|{\cal E}_i|$, and since $|{\cal Y}_i| \le \frac{\varepsilon}{100}|{\cal E}_i|$ we get that 
$$\sum_{j \in {\cal Y}_i}{\| D_j - A_j^{tmp}\|_1} \le 4 \cdot 2^{-i}M|{\cal Y}_i| \le \frac{4\varepsilon}{100} \cdot 2^{-i}M|{\cal E}_i| < \frac{4\varepsilon}{100} \sum_{j \in {\cal E}_i}{\varrho(D_j)} \le \frac{8\varepsilon}{100} \sum_{j \in {\cal E}_i}{\|D_j\|_1}\;.$$
We can therefore conclude that
$$\sum_{i \in [\log(m/\varepsilon)]}{\sum_{j \in {\cal Y}_i}{\| D_j - A_j^{tmp}\|_1}} \le \frac{8\varepsilon}{100} \sum_{j \in [m]}{\|D_j\|_1} \le O(\varepsilon)\veps\;,$$
where the last inequality is due to the fact that $\|D\|_1 = \|A - A^{init}\|_1 \le O(1)\veps$.
\end{proof}

Lemma~\ref{l:1PlusEps} now directly follows from Claims~\ref{c:sumDense}, \ref{c:sumSparse}. 
The proof of Theorem~\ref{th:main} is now complete.

\section{Adaptivity is Necessary Even for Noise-Free Signals} \label{s:nonAdapt}
To prove Theorem~\ref{t:nonAdaptiveLower}, we first note that every non-adaptive scheme $\ALG$ to the problem of reconstructing a $km$-sparse matrix $A \in \mathbb{R}^{n \times m}$ can be viewed as the concatenation of two algorithms. $\ALG^s$ constructs sensing matrices $S_1,\ldots,S_m$, and $\ALG^r$ is given the measurements $S_1A_1,\ldots,S_mA_m$, and recovers $A$. For every $j \in [m]$, let $t_j$ denote the number of rows of (i.e. measurements performed by) $S_j$. The total number of linear measurements performed by the scheme is therefore $t_{\ALG} \eqdef \sum_{j \in [m]}{t_j}$. 

Assume that for every $A \in \mathbb{R}^{n \times m}$, $\ALG$ reconstructs $A$ with success probability $\ge \frac{1}{2}$. Since $\ALG$ is non-adaptive, the number of measurements $r_{\ALG}$ depends only on $k,m,n$.

By Yao's minimax principle, it suffices to show a distribution ${\cal D}$ over 
matrices in $\mathbb{R}^{n \times m}$ such that for every deterministic algorithm $\ALG_{det}$, if $\Pr_{A \sim {\cal D}}[\ALG_{det} \; succeeds] \ge \frac{1}{2}$, then $t_{\ALG_{det}} \ge \Omega(mn)$. 
Consider a matrix $A$ constructed as follows. Choose uniformly at random $i^* \in [m]$. For every $j \in [m] \setminus \{i^*\}$, let $A_j = (1,0,0,\ldots,0)^t$, and let $A_{i^*} = (x_1,\ldots,x_n)^t$, where $x_1,\ldots,x_n \sim N(0,1)$ i.i.d. 
$\ALG_{det}^s$ constructs sensing matrices $S_1,\ldots,S_m$. The following lemma implies Theorem~\ref{t:nonAdaptiveLower}. 
\begin{lemma}\label{l:recProb}
Fix $j \in [m]$, and assume that $t_j \le n-1$, then conditioned on $i^*=j$, $\ALG_{det}^r$ fails to recover $A_j$ with probability $1$.
\end{lemma}

\begin{proof}
Conditioned on $i^*=j$, the distribution of $A_j$ is independent of $\{A_i\}_{i \ne j}$. We can therefore analyze the success probability of $\ALG_{det}^r$ on $A_j$ as if $\ALG_{det}^r$ receives only $(S_j,S_jA_j)$ as input, and attempts to recover $A_j$.
Denote $U = \Ker(S_j)$ and $\nu = \dim U$, then since $S_j$ is underdetermined, $\nu \ge 1$, and there is an orthonormal basis $u_1,\ldots,u_{n}$ to $\mathbb{R}^{n}$ satisfying that $u_1,\ldots,u_\nu$ is a basis of $U$, and $u_{\nu+1},\ldots,u_{n}$ is a basis of $U^\bot$. For every $y \in \Img(S_j)$ there exists a unique $x' \in U^\bot$ satisfying $S_jx = y$. 
Since $\ALG_{det}^r$ is deterministic, there exists a unique $x \in \mathbb{R}^{n}$ such that $\ALG_{det}^r$ returns $x$ when invoked on $(S_j,y)$. 

Denote $A_j = \sum_{\ell \in [n]}{y_\ell u_\ell}$, then since $\{u_\ell\}_{\ell \in [n]}$ is orthonormal, then $y_1,\ldots,y_{n} \sim N(0,1)$ i.i.d. Following the above discussion, the value of $S_jA_j$ is independent of $y_1,\ldots,y_\nu$. Moreover, for every $y_{\nu+1},\ldots,y_{n} \in \mathbb{R}$, there exist unique values $y^*_1,\ldots,y^*_\nu$ such that $\ALG_{det}^r(S_j,S_jA_j)$ is correct if and only if $y_\ell = y^*_\ell$ for all $\ell \in [\nu]$. Therefore,
$$\Pr[\ALG_{det}^r \; fails | i^* = j] = \int_{y_{\nu+1},\ldots,y_{n}\in \mathbb{R}} {\int_{(y_1,\ldots,y_\nu) \ne (y_1^*,\ldots,y_\nu^*)}{\left(\prod_{\ell \in [n+1]}{f(y_\ell)}\right)}dy_1\ldots dy_{n}} = 1 \;,$$
where $f$ is the pdf of the standard normal distribution.
\end{proof}
\begin{proof}[Proof of Theorem~\ref{t:nonAdaptiveLower}]
Denote $J \eqdef \{j \in [m] : t_j \ge n\}$. From Lemma~\ref{l:recProb}, whenever $j \notin J$, $\Pr[\ALG_{det}^r \; recovers \; A_j| i^* = j] = 0$.
Since $\ALG_{det}$ reconstructs $A$ with probability $\ge \frac{1}{2}$, then 
$$\frac{1}{2} \le \Pr[\ALG_{det}^r \; recovers \; A] = \Pr[\ALG_{det}^r \; recovers \; A| i^* \in J] \Pr[i^* \in J] \le \Pr[i^* \in J] \;.$$
Therefore,
$$t_{\ALG_{det}} = \sum_{j \in [m]}{t_j} \ge \sum_{j \in [m] : t_j \ge n}{t_j} \ge n \cdot \frac{m}{2} \ge \Omega(mn) \;.$$
\end{proof}

\section{Future Directions}
\label{sec:conclusion}
A natural important problem that arises from the main results, 
and is not completely resolved within the context of the paper 
is obtaining tight bounds on the number of adaptive 
rounds needed for a batch sparse recovery scheme that performs $\tilde{O}(km)$ linear measurements. 
A more refined question asks for the correct tradeoff between 
the number of adaptive rounds and number of linear measurements.

We showed in Theorem~\ref{t:nonAdaptiveLower} that 
every random non-adaptive scheme for batch sparse recovery 
must perform $\Omega(mn)$ measurements for every $A$ and $k$ with high probability, 
even for the case $\veps = 0$. 
As it turns out, for the restricted case $\veps=0$, two adaptive rounds 
are indeed enough to get an optimal number of linear measurements. 
Indeed, in the first round the algorithm estimates,
up to a constant factor, 
the ``correct'' sparsity bound of each column, i.e., $\|A_j\|_0$,
which can be done by performing $O(\log n)$ 
linear measurements on every column \cite{KNW10}. 
In the second round of measurements, 
the algorithm employs stable sparse recovery (e.g. Theorem~\ref{th:compNoise}) 
to reconstruct the unknown entries 
while using the correct sparsity bound, up to a constant factor. 
The total number of measurements is therefore $O(mk \log n)$, 
which is optimal by known sparse recovery lower bounds.

For arbitrary values of $\veps \in (0,1)$, however, 
the scheme presented in this paper uses $O(\log m)$ adaptive rounds of measurements 
in order to bound the number of measurements by $\tilde{O}(km)$.
We suspect that a doubling approach similar to that in Algorithm~\ref{alg:noise} is, 
in a sense, required and therefore every batch recovery scheme that performs $\tilde{O}(km)$ 
measurements must perform $\Omega(\log m)$ adaptive rounds.

A good starting point for studying this question lies in a ``noise-capped'' version of stable sparse recovery, 
where the input is $x\in\R^n$ as well as an intended noise level $\veps$ 
instead of an intended sparsity bound $k\in[n]$.
The goal is to recover $x' \in \R^n$ that satisfies $\|x - x'\|_1 \le O(\veps)$
using a small number of linear measurements to $x$,
where small is with respect to the minimal sparsity $k$ 
needed to approximate $x$ with relative error $\veps$.
Note that while each of $k,\veps$ completely determines the other,
which of them is given explicitly does matter algorithmically.
A straightforward doubling scheme attains the optimal number of measurements 
in $O(\log k)$ adaptive rounds. 
Moreover, it can be shown, similarly to the proof of Theorem~\ref{t:nonAdaptiveLower} 
that every non-adaptive scheme must use $\Omega(n)$ measurements.

While noise-capped sparse recovery appears to be a straightforward, or even na{\"i}ve problem, finding the correct tradeoff between the number of linear measurements even for this simple variant proves highly non-trivial.
Specifically, known techniques previously used to prove sparse recovery lower bounds \cite{BIPW10,PW12}, and specifically adaptive sparse recovery lower bounds \cite{PW13}, do not extend to this problem.


\newcommand{\etalchar}[1]{$^{#1}$}

\end{document}